\documentclass{amsart}
\usepackage{enumerate}
\usepackage{xcolor}
\usepackage{tikz}
\renewcommand{\subset}{\subseteq}
\usetikzlibrary{decorations.pathreplacing}
\usetikzlibrary{decorations.pathmorphing}
\tikzstyle{stretchy}=[decorate, decoration={snake, segment length=.2cm, amplitude=.05cm}]
\tikzstyle{medge}=[line width = 2pt]
\tikzstyle{vertex}=[inner sep = 0pt, minimum width=4pt, draw=black, fill=black, shape=circle]
\tikzstyle{whitevert}=[inner sep = 0pt, minimum width=4pt, draw=black, fill=white, shape=circle]
\tikzstyle{squarevert}=[inner sep = 0pt, minimum width=4pt, minimum height=4pt, fill=white, shape=rectangle, draw=black, thick]
\tikzstyle{trivert}=[inner sep = 0pt, minimum width=6pt, minimum height=6pt, fill=gray!50!white, shape=regular polygon,regular polygon sides=3, draw=black, thick]
\tikzstyle{pentvert}=[inner sep = 0pt, minimum width=6pt, minimum height=6pt, fill=gray!15!white, shape=regular polygon,regular polygon sides=5, draw=black, thick]

\newcommand{\gpoint}[2]{\node[style=vertex, label=#1:$#2$]}
\newcommand{\vpoint}[2]{\node[style=squarevert, label=#1:$#2$]}
\newcommand{\wpoint}[2]{\node[style=whitevert, label=#1:$#2$]}
\newcommand{\bpoint}[1]{\gpoint{below}{#1}}
\newcommand{\apoint}[1]{\gpoint{above}{#1}}
\newcommand{\lpoint}[1]{\gpoint{left}{#1}}
\newcommand{\rpoint}[1]{\gpoint{right}{#1}}
\newcommand{\bvoint}[1]{\vpoint{below}{#1}}
\newcommand{\bwoint}[1]{\wpoint{below}{#1}}
\newcommand{\avoint}[1]{\vpoint{above}{#1}}
\newcommand{\awoint}[1]{\wpoint{above}{#1}}

\newcommand{\sizeof}[1]{\left\lvert{#1}\right\rvert}
\newtheorem{proposition}{Proposition}[section]

\newtheorem{theorem}[proposition]{Theorem}
\newtheorem{lemma}[proposition]{Lemma}
\newtheorem{corollary}[proposition]{Corollary}
\theoremstyle{remark}
\newtheorem{question}[proposition]{Question}

\newcommand{\DM}{DM}
\newcommand{\SDM}{SDM}
\theoremstyle{definition}

\newtheorem{definition}[proposition]{Definition}
\title{Complexity of a Disjoint Matching Problem on Bipartite Graphs}
\author{Gregory J.~Puleo}

\newcommand{\st}{\colon\,}
\begin{document}
\begin{abstract}
  We consider the following question: given an $(X,Y)$-bigraph $G$ and
  a set $S \subset X$, does $G$ contain two disjoint matchings $M_1$
  and $M_2$ such that $M_1$ saturates $X$ and $M_2$ saturates $S$?
  When $\sizeof{S}\geq\sizeof{X}-1$, this question is solvable by
  finding an appropriate factor of the graph. In contrast, we show
  that when $S$ is allowed to be an arbitrary subset of $X$, the
  problem is NP-hard.
\end{abstract}
\maketitle
\section{Introduction}
A \emph{matching} in a graph $G$ is a set of pairwise disjoint edges.
A matching \emph{covers} a vertex $v \in V(G)$ if $v$ lies in some
edge of the matching, and a matching \emph{saturates} a set $S \subset
V(G)$ if it covers every vertex of $S$.

An $(X,Y)$-bigraph is a bipartite graph with partite sets $X$ and
$Y$. The fundamental result of matching theory is Hall's
Theorem~\cite{HallsTheorem}, which states that an $(X,Y)$-bigraph contains a
matching that saturates $X$ if and only if $\sizeof{N(S)} \geq
\sizeof{S}$ for all $S \subset X$. While Hall's Theorem does not
immediately suggest an efficient algorithm for finding a maximum
matching, such algorithms have been discovered and are well-known
\cite{bipmatch1, bipmatch2}.

A natural way to extend Hall's Theorem is to ask for necessary and
sufficient conditions under which \emph{multiple} disjoint matchings
can be found. This approach was taken by Lebensold, who obtained
the following generalization of Hall's~Theorem.
\begin{theorem}[Lebensold~\cite{Lebensold}]\label{thm:leben}
  An $(X,Y)$-bigraph has $k$ disjoint matchings, each saturating $X$,
  if and only if
  \begin{equation}
    \label{eq:lebensold}\sum_{y \in Y}\min\{k, \sizeof{N(y) \cap S}\} \geq k\sizeof{S} 
  \end{equation}
  for all $S \subset X$.
\end{theorem}
When $k=1$, the left side of \eqref{eq:lebensold} is just
$\sizeof{N(S)}$, so Theorem~\ref{thm:leben} contains Hall's~Theorem as
a special case. As observed by Brualdi, Theorem~\ref{thm:leben} is
equivalent to a theorem of Fulkerson~\cite{Fulkerson} about disjoint
permutations of $0,1$-matrices. Theorem~\ref{thm:leben} is also a
special case of Lovasz's $(g,f)$-factor theorem~\cite{Lovasz-gf}.
Like Hall's Theorem, Theorem~\ref{thm:leben} does not immediately
suggest an efficient algorithm, but efficient algorithms exist for
solving the $(g,f)$-factor problem~\cite{Gabow}, and these algorithms
can be applied to find the desired $k$ disjoint matchings. We discuss
the algorithmic aspects further in Section~\ref{sec:easycase}.

A different extension was considered by Frieze~\cite{frieze}, who
considered the following problem:
\begin{quote}
  \textbf{Disjoint Matchings (DM)}\\
  \textbf{Input:} Two $(X,Y)$-bigraphs $G_1$, $G_2$ on the same vertex set. \\
  \textbf{Question:} Are there matchings $M_1 \subset G_1$, $M_2 \subset G_2$
  such that $M_1 \cap M_2 = \emptyset$ and each $M_i$ saturates $X$?
\end{quote}
When $G_1=G_2$, this problem is just the $k=2$ case of the problem
considered by Lebensold, and is therefore polynomially solvable. On
the other hand, Frieze proved that the Disjoint Matchings problem is
NP-hard in general.

In this paper, we consider the following disjoint-matching problem, which
can be naturally viewed as a restricted case of the Disjoint Matchings problem:
\begin{quote}
  \textbf{Single-Graph Disjoint Matchings (SDM)}\\
  \textbf{Input:} An $(X,Y)$-bigraph $G$ and a vertex set $S \subset X$. \\
  \textbf{Question:} Are there matchings $M_1, M_2 \subset G$ such that
  $M_1 \cap M_2 = \emptyset$, $M_1$ saturates $X$, and $M_2$ saturates $S$?
\end{quote}
We call such a pair $(M_1, M_2)$ an \emph{$S$-pair}.  When $S=X$, this
problem is also equivalent to the $k=2$ case of Lebensold's problem.
The problem \SDM{} is similar to a problem considered by Kamalian and
Mkrtchyan~\cite{kamalian}, who proved that the following problem is
NP-hard:
\begin{quote}
  \textbf{Residual Matching}\\
  \textbf{Input:} An $(X,Y)$-bigraph $G$ and a nonnegative integer $k$. \\
  \textbf{Question:} Are there matchings $M_1, M_2 \subset G$ such that
  $M_1 \cap M_2 = \emptyset$, $M_1$ is a maximum matching, and $\sizeof{M_2} \geq k$?
\end{quote}
When $G$ has a perfect matching, we can think of the Residual Matching
problem as asking whether there is \emph{some} $S \subset X$ with
$\sizeof{S} = k$ such that $G$ has an $S$-pair. In contrast, the
\SDM{} problem asks whether some \emph{particular} $S$ admits an
$S$-pair. Since $k$ is part of the input to the Residual Matching problem,
it is \emph{a priori} possible that \SDM{} could be polynomially solvable
while the Residual Matching problem is NP-hard, since one might need to
check exponentially many candidate sets $S$.

In Section~\ref{sec:reduce}, we give a quick reduction from \SDM{} to
\DM{}, justifying the view of \SDM{} as a special case of \DM{}, and
in Section~\ref{sec:twomatch} we show that \SDM{} is NP-hard, thereby
strengthening Frieze's result. In Section~\ref{sec:easycase} we show
that \SDM{} is polynomially solvable under the additional restriction
$\sizeof{S} \geq \sizeof{X}-1$.

\section{Reducing \SDM{} to \DM{}}\label{sec:reduce}
In this section, we show that any instance of \SDM{} with $\sizeof{S} < \sizeof{X}-1$
reduces naturally to an instance of \DM{}. Since \SDM{}-instances with $\sizeof{S} \geq \sizeof{X}-1$
are polynomially solvable, as we show in Section~\ref{sec:easycase}, this justifies the claim 
that \SDM{} is a special case of \DM{}.
\begin{theorem}
  Let $G$ be an $(X,Y)$-bigraph and let $S \subset V(G)$ with $\sizeof{S} < \sizeof{X}-1$. Construct
  graphs $G_1, G_2$ as follows:
  \begin{align*}
    V(G_1) &= V(G_2) = V(G), \\
    E(G_1) &= E(G), \\
    E(G_2) &= E(G) \cup \{xy \st x \in X-S,\ y \in Y\}.
  \end{align*}
  The graph $G$ has an $S$-pair if and only if there are disjoint
  matchings $M_1, M_2$ contained in $G_1, G_2$ respectively, each
  saturating $X$.
\end{theorem}
\begin{proof}
  If $\sizeof{Y} < \sizeof{X}$, then it is clear that $G$ has no
  $S$-pair and that $G_1, G_2$ do not have perfect matchings, so
  assume that $\sizeof{Y} \geq \sizeof{X}$.

  First suppose that $M_1, M_2$ are disjoint matchings contained in
  $G_1, G_2$ respectively, each saturating $X$. Let $M'_1 = M_1$ and
  let $M'_2 = \{e \in M_2 \st e \cap X \subset S\}$.  It is clear that
  $(M'_1, M'_2)$ is an $S$-pair.

  Now suppose that we are given an $S$-pair $(M'_1, M'_2)$. In order
  to obtain the matchings $M_1, M_2$ in $G_1, G_2$ as needed, we need
  to enlarge $M'_2$ so that it saturates all of $X$, rather than only
  saturating $S$. Let $Y' = \{y \in Y \st y \notin V(M'_2)\}$, and let
  $H = G_2[(X-S) \cup Y'] - M'_1$.

  We claim that $H$ has a matching that saturates $X-S$, and prove
  this by verifying Hall's Condition. Let any $X_0 \subset X-S$ be
  given.  If $\sizeof{X_0} = 1$, say $X_0 = \{x_0\}$, then $N_H(X_0)$
  contains all of $Y'$ except possibly the mate of $x_0$ in $M_1$.
  Hence
  \[ \sizeof{N_H(X_0)} \geq \sizeof{Y'}-1 = \sizeof{Y} - \sizeof{S} - 1 \geq \sizeof{X} - \sizeof{S} - 1 \geq 1 = \sizeof{X_0}, \]
  as desired. On the other hand, if $\sizeof{X_0} \geq 2$, then $N_H(X_0)$
  contains all of $Y'$, so that
  \[ \sizeof{N_H(X_0)} = \sizeof{Y'} = \sizeof{Y} - \sizeof{S} \geq
  \sizeof{X} - \sizeof{S} \geq \sizeof{X_0}. \] Hence Hall's Condition
  holds for $H$. Now let $M$ be a perfect matching in $H$, let $M_1 =
  M'_1$, and let $M_2 = M'_2 \cup M$. By construction, $M_2$ is a
  matching in $G_2$ that saturates $X$. It is clear that $M_1 \cap M_2 =
  \emptyset$, since the edges in $M'_1$ were omitted from $H$. Hence $M_1$
  and $M_2$ are as desired.
\end{proof}

\section{Finding Two Matchings is NP-Hard}\label{sec:twomatch}
Given an instance $(G,S)$ of \SDM{}, we call a pair of matchings
$(M_1, M_2)$ satisfying the desired condition an \emph{$S$-pair}. When
$G'$ is a subgraph of $G$ and $S' = S \cap V(G')$, we say that an
$S$-pair $(M_1, M_2)$ \emph{contains} an $S'$-pair $(M'_1, M'_2)$ if
$M'_1 \subset M_1$ and $M'_2 \subset M_2$.

We prove that \SDM{} is NP-hard via a reduction from 3SAT.  Let $c_1,
\ldots, c_s$ be the clauses and $\theta_1, \ldots, \theta_t$ be the
variables of an arbitrary 3SAT instance. We define a graph $G$ as
follows.

For each variable $\theta_i$, let $H_i$ be a copy of the cycle $C_{4s}$, with vertices
$v_{i,1}, \ldots, v_{i,4s}$ written in order. Define
\begin{align*}
  X_i &= \{ v_{i,j} \st \text{$j$ is even} \}, \\
  S_i &= \{ v_{i,j} \st \text{$j \equiv 2 \pmod{4}$} \}.
\end{align*}
Since $H_i$ is an even cycle, it has exactly two perfect matchings,
one containing the edge $v_{i,1}v_{i,2}$ and the other containing the
edge $v_{i,2}v_{i,3}$. In an $S_i$-pair $(M_1, M_2)$ for $H_i$, we
have $v_{i,1}v_{i,2} \in M_1$ if and only if $v_{i,2}v_{i,3} \in M_2$,
and the same argument holds for the other vertices of $S_i$. Thus,
$H_i$ has only two possible $S_i$-pairs, illustrated in
Figure~\ref{fig:hipairs}. We call these pairs the \emph{true pair} and
\emph{false pair} for $H_i$.
\begin{figure}
  \centering
  \begin{tabular}{cc}
    \begin{tikzpicture}[scale=1.5]
      \apoint{v_{i,1}} (v1) at (90 : 1cm) {};
      \avoint{v_{i,2}} (v2) at (45 : 1cm) {};
      \rpoint{v_{i,3}} (v3) at (0 : 1cm) {};
      \bwoint{v_{i,4}} (v4) at (-45 : 1cm) {};
      \bpoint{v_{i,5}} (v5) at (-90 : 1cm) {};
      \bvoint{v_{i,6}} (v6) at (-135 : 1cm) {};
      \lpoint{v_{i,7}} (v7) at (-180 : 1cm) {};
      \awoint{v_{i,8}} (v8) at (-225 : 1cm) {};
      \draw[medge] (v1) -- (v2);
      \draw[medge] (v3) -- (v4);
      \draw[medge] (v5) -- (v6);
      \draw[medge] (v7) -- (v8);
      \draw[stretchy] (v2) -- (v3);
      \draw[stretchy] (v6) -- (v7);
      \draw (v1) -- (v8);
      \draw (v4) -- (v5);
    \end{tikzpicture}
    &\begin{tikzpicture}[scale=1.5]
      \apoint{v_{i,1}} (v1) at (90 : 1cm) {};
      \avoint{v_{i,2}} (v2) at (45 : 1cm) {};
      \rpoint{v_{i,3}} (v3) at (0 : 1cm) {};
      \bwoint{v_{i,4}} (v4) at (-45 : 1cm) {};
      \bpoint{v_{i,5}} (v5) at (-90 : 1cm) {};
      \bvoint{v_{i,6}} (v6) at (-135 : 1cm) {};
      \lpoint{v_{i,7}} (v7) at (-180 : 1cm) {};
      \awoint{v_{i,8}} (v8) at (-225 : 1cm) {};
      \draw[medge] (v2) -- (v3);
      \draw[medge] (v4) -- (v5);
      \draw[medge] (v6) -- (v7);
      \draw[medge] (v8) -- (v1);
      \draw[stretchy] (v2) -- (v1);
      \draw[stretchy] (v6) -- (v5);
      \draw (v7) -- (v8);
      \draw (v3) -- (v4);
    \end{tikzpicture}\\
    True pair.&False pair.
  \end{tabular}
  \caption{True and false pairs for $H_i$ in the case $s=2$. White vertices lie in $X_i$; black vertices lie in $Y_i$; square vertices
  lie in $S_i$. Thick lines denote edges in $M_1$, wavy lines denote edges in $M_2$.}
  \label{fig:hipairs}
\end{figure}
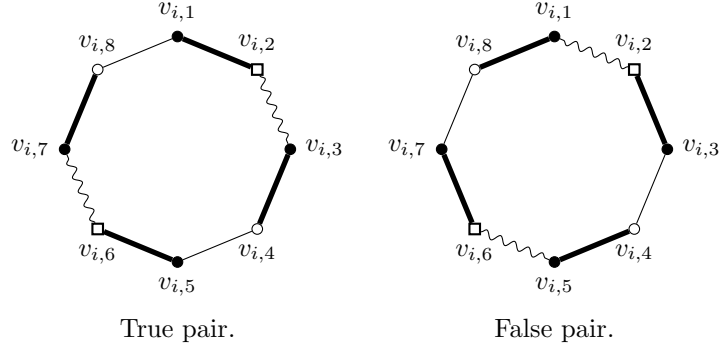

In the full graph $G$, we will not add any new edges incident to the vertices of $X_i$, so
it will still be the case that any $S$-pair in the full graph induces either the true pair
or the false pair in $H_i$. We use these pairs to encode the truth values of the corresponding
3SAT-variables.

For each clause $c_k$, let $L_k$ be a copy of $K_2$, with vertices $w_k, z_k$.
Let $G = \big(\bigcup_j H_j\big) \cup \big(\bigcup_k L_k\big)$. Add edges to $G$
as follows: if the variable $\theta_i$ appears positively in the clause $c_k$, add an
edge from $w_k$ to $v_{i,4k-3}$, and if the variable $\theta_i$ appears negatively in the
clause $c_k$, add an edge from $w_k$ to $v_{i,4k-1}$.

Let $X = \bigcup_j (X_j \cup \{w_j\})$, and let
$Y = V(G)-X$. Observe that $(X,Y)$ is a bipartition of $V(G)$. Let
$S = \left(\bigcup S_j\right) \cup \bigcup\{w_j\}$.
\begin{lemma}
  $G$ has an $S$-pair if and only if the given 3SAT instance is
  satisfiable.
\end{lemma}
\begin{proof}
  Let $(M_1, M_2)$ be an $S$-pair. We show that the 3SAT instance
  is satisfiable. 

  For any variable $\theta_i$, the vertices of $X \cap H_i$ have
  neighborhoods contained in $H_i$. Hence, $(M_1, M_2)$ contains an
  $S_i$-pair, and in particular contains either the true pair or the
  false pair for $H_i$.  Construct an assignment by setting each
  variable $\theta_i$ to be true if $(M_1, M_2)$ contains the true pair
  for $H_i$ and false otherwise. We claim that this is a satisfying
  assignment.

  Consider any clause $c_k$. Since $M_1$ is a perfect matching and $w_k$
  is the only neighbor of $z_k$, we have $w_kz_k \in M_1$. Since $w_k \in S$,
  some edge $w_kv_{i,4k-3}$ or $w_kv_{i,4k-1}$ lies in $M_2$.

  If $w_kv_{i,4k-1} \in M_2$, then $v_{i,4k-2}v_{i,4k-2} \notin M_2$, so the
  given $S$-pair contains the false pair for $H_i$. Since $w_kv_{i,4k-1} \in E(G)$,
  the clause $c_k$ contains a negative instance of $\theta_i$, so the constructed
  assignment satisfies the clause $c_k$. On the other hand, if $w_kv_{i,4k-3} \in M_2$,
  then the given $S$-pair contains the true pair for $H_i$ and $\theta_i$ appears
  postively in $w_k$, so we again see that $w_k$ is satisfied.

  Conversely, suppose that the 3SAT problem has a satisfying
  assignment.  Consider the pair of matchings $(M_1, M_2)$ in $G$
  obtained as follows. For each variable $i$, add the true pair for
  each $H_i$ where $\theta_i$ is true and the false pair for each $H_i$
  where $\theta_i$ is false. For each clause $c_k$, add the edge $w_kz_k$ to $M_1$.
  Choose some variable $\theta_i$ that satisfies the clause $c_k$. If $\theta_i$ is true,
  add the edge $w_kv_{i,4k-3}$ to $M_2$, and otherwise add $w_kv_{i,4k-1}$ to $M_2$.
  It is straightforward to check that this is an $S$-pair for $G$.  
\end{proof}
\begin{corollary}
  \SDM{} is NP-hard.
\end{corollary}
Viewing \SDM{} as a special case of \DM{} as in Section~\ref{sec:reduce}, we obtain
the following NP-hardness result for \DM{}.
\begin{corollary}
  \DM{} is NP-hard, even when restricted to instances for which $E(G_1) \subset E(G_2)$.
\end{corollary}
\section{An Algorithm for the Case $\sizeof{S} \geq \sizeof{X}-1$}\label{sec:easycase}
In this section, we provide a polynomial-time algorithm for solving
\SDM{} in the special case $\sizeof{S} \geq \sizeof{X}-1$. Our
algorithm requires the notion of a \emph{$(g,f)$-factor} as well
as the notion of \emph{edge coloring}.
\begin{definition}
  If $G$ is a graph and $g$ and $f$ are functions from $V(G)$ into the
  nonnegative integers, a \emph{$(g,f)$-factor} is a subgraph $H
  \subset G$ such that $g(v) \leq d_H(v) \leq f(v)$ for all $v \in
  V(G)$.
\end{definition}
Lovasz~\cite{Lovasz-gf} gave a Hall-like condition for a graph to have
a $g,f$-factor, and polynomial-time algorithms are known for
determining whether such a factor exists (for example,
\cite{Gabow}). In the bipartite case we are considering here, the
problem of determining whether such a factor exists can also be
reduced to a feasible-flow problem.
\begin{definition}
  For a nonnegative integer $k$, a \emph{$k$-edge coloring} of a graph
  $G$ is a function $f : E(G) \to \{1,\ldots,k\}$ such that $f(e_1)
  \neq f(e_2)$ whenever $e_1, e_2$ are distinct edges sharing an
  endpoint. The \emph{edge-chromatic number} of $G$, written
  $\chi'(G)$, is the smallest integer $k$ such that $G$ has a
  $k$-edge-coloring.
\end{definition}
\begin{theorem}[K\"onig's line-coloring theorem~\cite{konig1916}]
  If $G$ is a bipartite graph, then $\chi'(G)=\Delta(G)$, where $\Delta(G)$
  is the maximum degree of $G$.
\end{theorem}
\begin{theorem}\label{thm:minus1}
  There is a polynomial-time algorithm to solve \SDM{} restricted to instances
  for which $\sizeof{S} \geq \sizeof{X}-1$.
\end{theorem}
\begin{proof} 
  To avoid triviality, assume that $\sizeof{X} > 1$. Define functions $f$ and $g$
  as follows.
  \begin{align*}
    f(v) &=
    \begin{cases}
      1, &\text{if $v\in X-S$,} \\
      2, &\text{otherwise}
    \end{cases}\\    
    g(v) &=
    \begin{cases}
      f(v),& \text{if $v \in X$,} \\
      0,& \text{otherwise.}
    \end{cases}
  \end{align*}
  We can check in polynomial time whether $G$ has a $(g,f)$-factor. On
  the other hand, any $(g,f)$-factor $H$ has maximum degree $2$, and
  thus satisfies $\chi'(H) = 2$, by K\"onig's line-coloring
  theorem. Since $d_H(v) = 2$ for all $v \in S$, any $2$-edge-coloring
  of $H$ uses colors $\{1,2\}$ at each vertex of
  $S$. Furthermore, if $X-S\neq \emptyset$, then by switching
  colors if necessary, we can assume that the vertex in $X-S$ has only
  $1$ as an incident color. Taking $M_1$ and $M_2$ to consist of the
  edges of color $1$ and $2$ respectively, we see that $(M_1, M_2)$ is
  an $S$-pair in $G$. Conversely, if $(M'_1, M'_2)$ is any $S$-pair in $G$,
  then $M'_1 \cup M'_2$ is a $(g,f)$-factor.

  Hence, $G$ has a $(g,f)$-factor if and only if $G$ has an $S$-pair,
  so checking for such a factor solves the problem in polynomial time.
\end{proof}
For any fixed $k$, the problem \SDM{} is polynomial-time solvable on
instances with $\sizeof{S} \leq k$: we can iterate over the
$O(\sizeof{Y}^k)$ possible choices for $M_2$, and for each possible
choice, check whether $G-M_2$ has a perfect matching $M_1$. Since the
reduction in Section~\ref{sec:twomatch} produces \SDM{} instances in
which $\sizeof{X-S}$ is arbitrarily large, Theorem~\ref{thm:minus1}
suggests that \SDM{} might also be polynomially solvable when
$\sizeof{S}$ is bounded less strongly from below. However, the trick
of using $(g,f)$-factors is no longer sufficient by itself to solve
the problem when $k > 1$.
\begin{question}
  For fixed $k > 1$, is there a polynomial-time algorithm to solve
  \SDM{} on instances with $\sizeof{S} \geq \sizeof{X}-k$?
\end{question}
\bibliographystyle{amsplain} \bibliography{biblio}
\end{document}